\documentclass{article}

\usepackage{amssymb,latexsym,amsmath,amsthm,hyperref}

\usepackage[margin=1.4in]{geometry}
\usepackage{graphicx}
\usepackage{parskip}

\usepackage{mathtools} 
\mathtoolsset{showonlyrefs,showmanualtags}

\makeatletter
\def\thm@space@setup{%
  \thm@preskip=\parskip \thm@postskip=0pt
}
\makeatother

\newtheorem{thm}{Theorem}[section]
\newtheorem{remark}{Remark}
\newtheorem{lemma}[thm]{Lemma}

\newtheorem{p}{Problem}
\newtheorem{q}{Question}
\newtheorem{conj}{Conjecture}

\newcommand{\bep}{\begin{prob}}
\newcommand{\ep}{\end{prob}}
\newcommand{\be}{\begin{equation}}
\newcommand{\ee}{\end{equation}}
\newcommand{\cD}{\mathcal{D}}

\newcommand{\pd}{\partial}

\newcommand{\RR}{\mathbb{R}}
\newcommand{\CC}{\mathbb{C}}

\newcommand{\e}{\varepsilon}

\newcommand{\hF}{\hat{F}}

\usepackage[colorinlistoftodos]{todonotes}

\begin{document}
\author{Nickolas Arustamyan, Christopher Cox, Erik Lundberg, \\ 
Sean Perry, Zvi Rosen 
\vspace{0.1in}
\\
{\it Florida Atlantic University}\\
{\it Department of Mathematical Sciences}\\
{\it 777 Glades Rd} \\
{\it Boca Raton, FL 33431}}
\title{On the Number of Equilibria Balancing Newtonian Point Masses with a Central Force}

\maketitle

\begin{abstract}
We consider the critical points (equilibria) of a planar potential generated by $n$ Newtonian point masses augmented with a quadratic term (such as arises from a centrifugal effect).  Particular cases of this problem have been considered previously in studies of the circular restricted $n$-body problem.  
We show that the number of equilibria is finite for a generic set of parameters, and we establish estimates for the number of equilibria.
We prove that the number of equilibria is bounded below by $n+1$, and we provide examples to show that this lower bound is sharp.  We prove an upper bound on the number of equilibria that grows exponentially in $n$.
In order to establish a lower bound on the maximum number of equilibria, we analyze a class of examples, referred to as ``ring configurations'', consisting of $n-1$ equal masses positioned at vertices of a regular polygon with an additional mass located at the center.  Previous numerical observations indicate that these configurations can produce as many as $5n-5$ equilibria.
We verify analytically that the ring configuration has at least $5n-5$ equilibria when the central mass is sufficiently small.
We conjecture that the maximum number of equilibria grows linearly with the number of point masses.
We also discuss some mathematical similarities to other equilibrium problems in mathematical physics, namely, Maxwell's problem from electrostatics and the image counting problem from gravitational lensing.
\end{abstract}

\section{Introduction}

In this paper we are interested in the following problem.

\begin{p}\label{prob:main}
Study the number of equilibria (critical points) $z = (x,y) \in \RR^2$ that the potential
\be\label{eq:pot}
F(z):=\frac{|z|^2}{2} + \sum_{i=1}^n \frac{m_i}{|z-z_i|} = \frac{x^2+y^2}{2} + \sum_{i=1}^n \frac{m_i}{\sqrt{(x-x_i)^2+(y-y_i)^2}}
\ee
can have, where $z_i \in \RR^2$, and $m_i>0$, in other words study the number of solutions of the system
\be\label{eq:main}
\begin{cases}
\displaystyle x-\sum_{i=1}^n \frac{x-x_i}{((x-x_i)^2+(y-y_i)^2)^{3/2}} =0 \\
\displaystyle
y-\sum_{i=1}^n \frac{y-y_i}{((x-x_i)^2+(y-y_i)^2)^{3/2}} =0 
\end{cases},
\ee
obtained by setting the gradient of \eqref{eq:pot}
to zero.
What is the maximum number $N(n)$ 
of equilibria for each $n$?
\end{p}

One source of motivation for this problem comes from studying the restricted $(n+1)$-body problem, where the dynamics of an object of negligibly small mass are studied amid the Newtonian force field of a \emph{central configuration}, a rigidly rotating configuration of $n$ primary masses (``primaries'') with masses $m_i$.
In a co-rotating frame of reference,
where the positions of the rigidly rotating primaries become stationed at points $z_i$, the centrifugal effect gives rise to an additional quadratic term, and we arrive at (\ref{eq:pot}) 
as the \emph{effective potential} (or ``ammended potential'' \cite{KuRoSm}) after an appropriate rescaling.
(In addition to the centrifugal effect, the dynamical problem must take into account the Coriolis effect, but at points of equilibrium the Coriolis term vanishes.)
Equilibrium solutions in the co-rotating frame of reference play a fundamental role in understanding the dynamical problem \cite{Mireles2}, \cite{Mireles3}, \cite{Mireles4}.
Such equilibria are also important for classifying \emph{central configurations with an inferior mass} \cite{inferior}.

It is only special configurations 
of the primaries that solve the $n$-body problem under rigid rotation, earning the designation of central configurations.
Hence, 
this specific application of Problem \ref{prob:main} to the circular restricted $n$-body problem goes hand in hand with
the difficult question 
of which choices of $z_i$ and $m_i$ 
correspond to valid central configurations---just showing finiteness (up to similarity) of the number of such central configurations would solve Smale's 6th problem for the 21st century \cite{Sm}.
However, it seems natural to consider Problem \ref{prob:main} in the generality with which it is stated, over the full range of the parameters $z_i \in \RR^2$ and $m_i > 0$, both for the sake of mathematical elegance and also because Problem \ref{prob:main} has other physical interpretations valid over the full range of parameter space.
For instance, in an electrostatic setting the negation of the potential \eqref{eq:pot} can represent that of anchored Coulomb point charges augmented by a quadratic confining potential.

Previous work on Problem \ref{prob:main} has been restricted to special cases emerging from the above-mentioned application in celestial mechanics.
In \cite{KuRoSm}, the authors use BKK (Bernstein-Khovanskii-Kushnirenko) theory to prove an upper bound on the number of equilibria for the case of $n=3$ point masses located at the vertices of an equilateral triangle.
A complete classification of the number of equilibria was provided in \cite{BaLe} using rigorous numerical analysis.
In \cite{Kalvouridis1998}, the author performs a numerical investigation of the number of equilibria when the point masses are located at vertices of a regular $(n-1)$-gon along with a point mass located at the center, and when the central mass is sufficiently small in comparison with the peripheral masses the system is conjectured to have $5n-5$ equilibria based on numerical computations.  We verify this analytically in Theorem \ref{thm:5n-5} below.


Problem \ref{prob:main}
is reminiscent of Maxwell's problem \cite{GaNoSh} in electrostatics to study the number of equilibria admitted by a point charge system
(in Maxwell's problem there is no confining potential, the points are generally in three-dimensional space, and the charges are not required to have a common sign).
In his 1891 treatise on electricity and magnetism \cite{Maxwell}, J. Maxwell conjectured a quadratic upper bound of $(n-1)^2$ for the number of equilibria.
In \cite{GaNoSh}, Gabrielov, Novikov, and Shapiro obtained (under a finiteness assumption which was shown to hold generically) upper bounds using Bernstein-Khovanski-Kushnirenko theory and lower bounds using Morse theory.
Their upper bound grows super-exponentially in the number of point masses and thus we are still far from a proof of Maxwell's conjecture.  The authors of \cite{GaNoSh} (see also \cite{Killian})  considered the following restricted planar version of Maxwell's problem that is even closer to Problem \ref{prob:main}.

\begin{p}[Maxwell's problem restricted to the plane]\label{prob:Max}
Study the number of critical points $z = (x,y) \in \RR^2$ of the electrostatic potential
\be\label{eq:Maxwell}
E(z):=\sum_{i=1}^n \frac{m_i}{|z-z_i|},
\ee
where $z_i \in \RR^2$, $m_i \in \RR$.
\end{p}
It is conjectured in \cite{GaNoSh} that the maximum number of equilibria for Problem \ref{prob:Max} is controlled by the complexity of an associated Voronoi complex and grows linearly with $n$.  An exponentially growing upper bound was obtained in \cite{Killian}.  We point out a slight improvement in the Concluding Remarks below.

Problem \ref{prob:main} is also reminiscent of
the following image counting problem \cite{KhNe2}, \cite{Pet10} from gravitational lensing where gravitationally lensed images of a single background source can be viewed as equilibria of a potential, referred to as the \emph{time delay function}. 
\begin{p}[gravitationally lensed images]\label{prob:gravlens}
Study the number of critical points $z = (x,y) \in \RR^2$ of the time delay function
\be\label{eq:gravlens}
T(z)=\frac{|z-w|^2}{2} - \sum_{i=1}^n m_i \log|z-z_i|,
\ee 
associated to a gravitational lens generated by $n$ point masses positioned at $z_i \in \RR^2$ with mass $m_i>0$ with a background source positioned at $w \in \RR^2$.
\end{p}
The key difference here is that Problem \ref{prob:gravlens} involves \emph{logarithmic} instead of Newtonian point mass potentials. This leads to a gradient field $\nabla T (z)$ that is harmonic, and moreover critical points of the time delay can be viewed as fixed points of an anti-holomorphic map.
Settling astronomer S. Rhie's conjecture \cite{Rhie2001}, D. Khavinson and G. Neumann \cite{KhNe} proved a sharp upper bound of $5n-5$ for the number of critical points of the time delay $T$.
Their proof uses a clever indirect method based on anti-holomorphic dynamics that was introduced in \cite{KhSw}.
This completely settles the problem of determining the maximum number of lensed images for single-plane lensing.  See \cite{Bleher}, \cite{SeteCMFT}, \cite{SeteGrav}, \cite{SetePert} for further progress in understanding Problem \ref{prob:gravlens}.

The analogous problem of determining the sharp upper bound for multi-plane lensing is still open and was discussed in \cite{Pet10}.  An upper bound has recently been established in \cite{Perry} by the fourth named author of the current paper.
A sharp lower bound that holds in the multiplane case was proved earlier by A.O. Petters in \cite{Pe2} using Morse theory.

\begin{remark}\label{rmk:ring}
Rhie constructed examples of $n$-point gravitational lenses attaining $5n-5$ images in \cite{Rhie}.  Her examples place the point masses at the same locations as the ``ring'' configurations from \cite{Kalvouridis1998} that we use in our proof of Theorem \ref{thm:5n-5}.
\end{remark}

In this paper, we will focus on Problem \ref{prob:main}.
We prove a generic finiteness result for the number of equilibria (see Theorem \ref{thm:genfin}), a sharp lower bound (see Theorem \ref{thm:LB}), and an exponentially growing upper bound (see Theorem \ref{thm:UB}).  We prove that the maximum number of equilibria is at least $5n-5$ (see Theorem \ref{thm:5n-5}), and
we conjecture that the maximum number of equilibria grows linearly with $n$ (see Concluding Remarks).

\section{Sharp Lower Bound for the Number of Equilibria}

In this section, we use Morse theory to prove a lower bound on the number of equilibria for Problem \ref{prob:main}.  First we establish a preliminary genericity result.



\subsection{Generic Finiteness}

In the case $n=1$, setting $z_1=0$ produces infinitely many equilibria (a circle centered at the origin).
It is an open problem to determine whether this is the only example with infinitely many equilibria.

Here we show, outside a set of parameters having measure zero, finiteness of the number of equilibria for Problem \ref{prob:main}.  The proof is based on Sard's Lemma.  An alternative algebraic proof of finiteness is given in the proof of Theorem \ref{thm:UB} below.

\begin{thm}\label{thm:genfin}
For a generic set of parameters (i.e., outside a set of measure zero), the potential function $F$ defined in \eqref{eq:pot} has finitely many critical points all of which are non-degenerate (i.e., the Hessian determinant is non-vanishing at each critical point).
\end{thm}

\begin{proof}
Computing the gradient, we write \eqref{eq:main} as
\be\label{eq:grad}
z - \sum_{i=1}^n m_i \frac{ z-z_i }{|z-z_i|^{3}} = 0.
\ee
Fixing the parameter $z_1$ let us use the change of variable $z=w-z_1$,
and new parameters $w_i=z_i+z_1$ for $i=2,3,...,n$.
Then \eqref{eq:grad} becomes
\be\label{eq:gradw}
\underbrace{w - m_1\frac{w}{|w|^3} - \sum_{i=2}^n m_i \frac{ w-w_i }{|w-w_i|^{3}}}_{\psi(w)} = z_1.
\ee
where $\psi : \RR^2\setminus \{0,w_2,...,w_n\} \to \RR^2$.
Note this change of variables describes a translation of the graph of the original potential function given in (\ref{eq:pot}) so that the origin corresponds to the singularity at $z_1$.
We may now rephrase the critical point condition (\ref{eq:main}) in order to apply Sard's Theorem; solutions $z$ of the critical point equations (\ref{eq:main}) correspond through the change of variables $w=z+z_1$ to preimages $w\in\psi^{-1}(z_1)$ of $z_1$ under the map $\psi$.  


Let 
$$V \subset\{(z_1,w_2,w_3,...,w_n,m_1,m_2,...,m_n)\in\RR^{2n} \times \RR_+^n\},$$ denote the set of parameters for which \eqref{eq:gradw} has infinitely many solutions,
where $\RR_+^{n}$ is the $n$-fold cartesian product of the non-negative reals.
Our goal is to show that $V$ has measure zero.
It suffices to show that each slice of $V$ in the direction of the $z_1$ coordinate (that is, any slice obtained by fixing values of the other parameters $(w_2,w_3,...,w_n,m_1,m_2,...,m_n)$)
has two-dimensional measure zero.


Fixing the values of $(w_2,w_3,...,w_n,m_1,m_2,...,m_n)$,
we note that $|\psi(w)|$ is bounded from below by a positive constant 
for $|w|$ large and $w$ near $w_i$. 
Hence, for $R>0$ sufficiently large and $\e>0$ sufficiently small, solutions to (\ref{eq:gradw}) are in the domain $\mathcal{D} \subset \RR^2$ defined by
\be \label{eq:disc}
\mathcal{D} = \mathring D_{R}(0) - D_{\varepsilon}(0) - \bigcup_{i=2}^n D_{\varepsilon}(w_i),
\ee
where $D_r(u)$ denotes the disc of radius $r$ centered at $u$. 

Solutions of equation \eqref{eq:gradw}
are preimages of $z_1$ under the smooth mapping $
\psi $ restricted to $\mathcal{D}$.
The number of preimages of a point $z_1$ is finite if $z_1$ is a regular value of $\psi$
\cite[p. 8]{Milnor}.

By Sard's Lemma \cite[p. 16]{Milnor}, the set of critical values of $\psi$ has measure zero, and this proves that each slice of $V$ in the direction of the $z_1$ coordinate has two-dimensional measure zero. This implies, for instance by Tonelli's theorem, that $V$ itself has measure zero.
\end{proof}

\subsection{Lower Bound for the Number of Equilibria}

\begin{thm}\label{thm:LB}
For a generic set of parameters, the number of critical points of $F$ is bounded from below by $n+1$.
\end{thm}


\bigskip

The proof of this theorem follows the method from \cite{Pe2} based on Morse theory in the setting of manifolds with boundary.
Let us first recall the Morse inequalities which, under suitable conditions, give a lower bound for the number of critical points of a smooth function in terms of the topology of its domain. Let $f:U\subset\RR^n \to \RR$ be smooth, $U$ open. A critical point of $f$ is called \textit{nondegenerate} if the Hessian determinant of $f$ is nonzero at that point. Such a function is called \textit{Morse} over $V\subseteq U$ if all critical points in $V$ are nondegenerate. Let $D\subset \RR^{n}$ be a
a domain whose closure is compact. If the boundary $\pd D$ of $D$ can be considered a smooth $(n-1)$-manifold, that is if $\pd D$ constitutes a locally-Euclidean Hausdorff space that can be fitted with a smooth differentiable structure, then $D$ is called a \textit{regular smooth domain}. 

Assume $f$ is Morse over a neighborhood of the closure of a regular smooth domain $D$, and satisfies that the gradient $\nabla f$ is outward-pointing at every point of $\pd D$. Let $N_i$ be the number of critical points of $f$ with \textit{index} $i$, that is whose Hessian has $i$ negative eigenvalues. Then the following are referred to as the \textit{Weak Morse inequalities}: $$N_i \geq B_i,$$ where $B_i$ is the $i^{th}$ Betti number of $D$. Hence the total number of critical points  $N=\sum N_i$ is bounded from below by 
\be\label{eq:Morse}
N \geq \sum B_i.
\ee

\begin{proof}[Proof of Theorem \ref{thm:LB}]
Let $F$ be defined as in (\ref{eq:pot}).
By Theorem \ref{thm:genfin}, for a generic choice of parameters we may assume $F$ has finitely many critical points, all non-degenerate. 
As in the proof of Theorem \ref{thm:genfin}, 
we let $\mathcal{D} \subset \RR^2$ denote a domain containing all critical points of $F$ obtained by taking a disk of sufficiently large radius $R$ and removing $n$ small disks with sufficiently small radius $\e>0$ centered at the points $z_i$, i.e.,
\be
\mathcal{D} = \mathring D_{R}(0) - \bigcup_{i=1}^n D_{\varepsilon}(z_i),
\ee
where $D_r(u)$ denotes the disc of radius $r$ centered at $u$. 

Then $\cD$ is a regular smooth domain, and $F$ is Morse over a neighborhood of the closure of $\mathcal{D}$. Moreover, for $\e>0$ sufficiently small and $R$ sufficiently large, the gradient $\nabla F$ is outward-pointing on $\pd \mathcal{D}$. Indeed, consider the boundary component $\{ |z-z_1| = \e \}$.  We have
$$
\nabla F = -\frac{z-z_1}{|z-z_1|^3} + \underbrace{z - \sum_{i=2}^n \frac{z-z_i}{|z-z_i|^3}}_{F_{1}(z)}.
$$
In a neighborhood of $z_1$, the term $F_1(z)$ is bounded in norm by a constant, while $-\frac{z-z_1}{|z-z_1|^3}$ points in the outward normal direction along each circle $|z - z_1|=\e$ of radius $\e>0$ centered at $z_1$ and has norm approaching infinity as $\e \rightarrow 0$.  Thus, for sufficiently small $\e$, $\nabla F$ has positive dot product with the outward pointing normal at each point along $|z - z_1|=\e$, i.e., $\nabla F$ is outward pointing.
Similarly, $\nabla F$ is outward pointing along each of the boundary components $\{|z-z_i|=\e\}$ if $\e>0$ is sufficiently small.  Along the large circle $\{|z|=R\}$, the dominant term in the gradient
$
\nabla F = z - \sum_{i=1}^n \frac{z-z_i}{|z-z_i|^3}
$
is $z$ which points in the outward normal direction and has norm approaching infinity as $R \rightarrow \infty$, while the remaining terms have norm approaching zero.  Thus, for $\e>0$ sufficiently small, $\nabla F$ is outward pointing along all boundary components.

Since $\mathcal{D}$ is
path-connected, $B_0(\mathcal{D})=1$. Since a disk is contractible, $D_R(0)$ has first Betti number zero.  Each removal of a small disk $D_\e(z_i)$ corresponds to the creation of an independent 1-cycle and thus $B_1(\mathcal{D})=n$. For larger values of $d$, we have $B_d(\mathcal{D})=0$. Thus the estimate \eqref{eq:Morse} based on the Morse inequalities gives the desired lower bound $N\geq n+1$ for the number $N$ of equilibria.
\end{proof}

\begin{remark}
With minor adaptations in the proof, Theorem \ref{thm:LB} holds for a more general class of potentials replacing the Newtonian potential $|z-z_i|^{-1}$ with a Riesz potential $|z-z_i|^{-\alpha}$, $\alpha>0$.
\end{remark}

\subsection{An Example Showing Sharpness of Theorem \ref{thm:LB}}

Here we provide an example where the lower bound in Theorem \ref{thm:LB} is attained, i.e., the number of equilibria is exactly $n+1$.


Consider $n$ masses at $(z_1,z_1),(z_2,z_2),...,(z_n,z_n)$ with mass $1$, and $0< z_1 < z_2 < \cdots < z_n$.
We will show that this configuration has exactly $n+1$ equilibria.
The equilibria are solutions of the system
\be\label{eq:Hpert}
\begin{cases}
\displaystyle F_x=x-\sum_{i=1}^n \frac{x-z_i}{((x-z_i)^2+(y-z_i)^2)^{3/2}} =0 \\
\displaystyle
F_y=y-\sum_{i=1}^n \frac{y-z_i}{((x-z_i)^2+(y-z_i)^2)^{3/2}} =0 
\end{cases},
\ee
Note that equilibrium points satisfy $F_x-F_y=0$.
We claim that this equation is only satisfied on the line $x-y=0$. Indeed, suppose $x-y \neq 0$.
We have
\begin{align*}
F_x-F_y&=x-y-\sum_{i=1}^n \frac{x-z_i-(y-z_i)}{((x-z_i)^2+(y-z_i)^2)^{3/2}} \\
&=x-y-(x-y)\sum_{i=1}^n
\frac{1}{((x-z_i)^2+(y-z_i)^2)^{3/2}}.
\end{align*}
So $F_x-F_y=0$ implies
\be\label{eq:FxFy}
1 - \sum_{i=1}^n \frac{1}{((x-z_i)^2+(y-z_i)^2)^{3/2}} = 0.
\ee

Since $x-y \neq 0$, $x$ and $y$ are not both zero.  Since $x$ and $y$ play symmetric roles in the problem, we may assume without loss of generality that $x \neq 0$. 


Now, consider $F_x/x$. We have
\begin{align}
\frac{F_x}{x} &= 1-\sum_{i=1}^n\frac{1-z_i/x}{((x-z_i)^2+(y-z_i)^2)^{3/2}} \\
&=\frac{1}{x}\sum_{i=1}^n\frac{-z_i}{((x-z_i)^2+(y-z_i)^2)^{3/2}},
\end{align}
where we have applied \eqref{eq:FxFy} in going from the first to the second line above.
Since $z_i>0$ for all $i$, we conclude that $F_x/x$ is nonzero.

As desired, this shows that if $x-y \neq 0$, then $(x,y)$ is not a point of equilibrium.

\bigskip

Now consider the line $y=x$. Note that on this line we have
\be
F_x(x,x)=F_y(x,x)=f(x):=-\sum_{i=1}^n \frac{x-z_i}{(2(x-z_i)^2)^{3/2}},
\ee 
so if $F_x=0$, then $F_y=0$ as well.
This reduces the problem to finding zeros of the univariate function $f(x)$.
We obtain, after differentiation,
$$\frac{df}{dx}(x)=1+\sum_{i=1}^n\frac{1}{\sqrt{2}((x-z_i)^2)^{3/2}}.$$

Thus, $\frac{df}{dx}$ is positive for all $x \neq z_i$, so $f$ is strictly increasing on each of the intervals $$(-\infty,z_1), (z_1,z_2),..., (z_n,\infty).$$ As such, $f$ has at most one zero on each interval. As there are $n+1$ intervals, there are at most $n+1$ equilibria, and by Theorem \ref{thm:LB} there are exactly $n+1$ equilibria.

\section{Estimates for the Maximum Number of Equilibria}

Let us next consider the maximum $N(n)$ of the number of equilibria taken over the generic set of parameters for which the number of equilibria is finite.
We provide some estimates for $N(n)$ in this section.  We first consider the problem of estimating $N(n)$ from above.

\subsection{An Upper Bound for the Maximum Number of Equilibria}

In \cite{KuRoSm}, the authors consider a special class of equilateral triangle examples, where they use BKK theory to obtain an upper bound of $196$.
The sharp upper bound of $10$ was obtained later \cite{BaLe} using a rigorous computational study 
of the bifurcation curve (a curve in the reduced parameter space separating cases according to number of equilibria).

The following upper bound grows exponentially in $n$.  We believe this is far from sharp, and it is even tempting to conjecture that $N(n)$ grows linearly in $n$.

\begin{thm}\label{thm:UB}
For a generic set of parameters, the maximum number $N(n)$ of equilibria for Problem \ref{prob:main} satisfies
\be 
N(n) \leq  4^{n+2}.
\ee 
\end{thm}

\begin{proof}
We associate to the original system of equations
\be\label{eq:original}
\begin{cases}
\displaystyle x-\sum_{i=1}^n \frac{x-x_i}{((x-x_i)^2+(y-y_i)^2)^{3/2}} =0 \\
\displaystyle
y-\sum_{i=1}^n \frac{y-y_i}{((x-x_i)^2+(y-y_i)^2)^{3/2}} =0 
\end{cases},
\ee
a system of \emph{polynomial} equations
by introducing the auxiliary variables
$$w_i=\sqrt{\frac{1}{(x-x_i)^2+(y-y_i)^2}},$$
and appending the additional equations
$
w_i^2((x-x_i)^2+(y-y_i)^2))-1=0
$
for $i=1,2,...,n$.
Then each solution $(x,y)$ of the original system \eqref{eq:original} gives rise to a solution $(x,y,w_1,w_2,...,w_n)$ of the system
\be\label{eq:append}
\begin{cases}
\displaystyle x - \sum_{i=1}^n m_i(x-x_i)w_i^3 = 0 \\
\displaystyle
y-\sum_{i=1}^n m_i(y-y_i)w_i^3  =0 \\
w_i^2((x-x_i)^2+(y-y_i)^2)-1=0, \quad i = 1,2,...,n
\end{cases}.
\ee
This is a system of $n+2$ polynomial equations in $n+2$ unknowns. 
We will estimate the number of complex solutions 
which gives an upper bound on the number of real solutions.
Each equation is of fourth degree, so if the number of solutions to the system of equations is finite, then by B\'ezout's theorem the number of solutions is at most $4^{n+2}$.
This gives the desired upper bound for the maximum number $N(n)$ of equilibria.


Hence, we just need to show that, for a generic set of parameters, the number of solutions of the system \eqref{eq:append} is indeed finite.

Let $\Delta \subset \RR^{3n}$ denote the set of parameters $m_i$, $x_i$, $y_i$ for which the system of equations \eqref{eq:append} is degenerate.  
Let 
$$\zeta = (x,y,w_1,w_2,...,w_n) \in \CC^{n+2},$$ and $$\eta = (m_1,m_2,...,m_n,x_1,y_1,x_2,y_2,...,x_n,y_n) \in \RR^{3n}$$ Then $\Delta$ is the projection onto the second factor (the space of parameters $\eta \in \RR^{3n}$) of the set
\be\label{eq:sysdeg}
\{ (\zeta,\eta) \in \CC^{n+2} \times \RR^{3n} : \hF(\zeta)=0, \; \det (D_\zeta \hF)=0 \},
\ee
where $\hF(\zeta)$ is the vector-valued function of left-hand-sides of equations from the system \eqref{eq:append}, and $\det(D_\zeta \hF)$  denotes the Jacobian determinant.
The set $\Delta$ is an algebraic variety.
It is the zero set of the polynomial $R \big\rvert_{\nu = \mu(\eta)}$, obtained by first computing the resultant $R$ for the system \eqref{eq:sysdeg} but with a general list $\nu$ of coefficients, and then making appropriate substitutions $\nu = \mu(\eta)$ corresponding to the polynomial dependence of the coefficients $\nu$ on the parameters $\eta$ (e.g., the coefficient of $w_i^3$ in the first equation is $-m_i x_i$).
For the fact that the resultant $R$ in general gives a defining polynomial for the discriminant of a system, we refer the reader to \cite[Ch. 13]{Gelfand}.

We are done if we show that $\Delta$ is a \emph{proper} subset of the parameter space $\RR^{3n}$ (this will imply that it has codimension at least one, since $\Delta$ is an algebraic variety).  
To show that $\Delta$ is proper, it is enough to exhibit one choice of parameters for which $\hF$ is non-degenerate.
We choose $x_i = 1, y_i = 1, m_i=0$.
Then the system \eqref{eq:sysdeg} becomes
\be\label{eq:append2}
\begin{cases}
\displaystyle x = 0 \\
\displaystyle
y  =0 \\
w_i^2((x-1)^2+(y-1)^2)-1=0, \quad i = 1,2,...,n, \\
\prod_{i=1}^n 2w_i ( (x-1)^2 + (y-1)^2) = 0.
\end{cases}
\ee
We can see that this system has no solutions by substituting $x=0, y=0$ from the first two equations into the remaining equations which gives the system of equations
$2w_i^2=1$ for $i=1,2,...,n$ and $\prod_{i=1}^n w_i =0$ which has no common solutions.
So the system \eqref{eq:append2} is non-degenerate for this choice of parameters, and this completes the verification that the number of solutions of \eqref{eq:append2} is finite for a generic set of parameters.
\end{proof}

\medskip

\subsection{Upper Bounds from Computational Algebra}

It is possible to improve the estimate in Theorem \ref{thm:UB} by considering
  particular small cases of the polynomial system appearing in \eqref{eq:append}.
  For instance, using the computational algebraic software {\tt Macaulay2} \cite{M2},
  we compute the {\em degree} of the corresponding zero-dimensional algebraic variety.
  This will bound the number of complex solutions, which in turn bounds the real solutions.
\begin{table}[h!]
  \begin{center}
    \caption{Upper bounds for $N(n)$.}
    \label{tab:table1}
    \begin{tabular}{c|r|r} 
       & Macaulay2 & $4^{n+2}$ \\
      $n$ &  &  \\
      \hline
      2 & 120  & 256 \\
      3 & 696 & 1024 \\
      4 & 3544 & 4096
    \end{tabular}
  \end{center}
\end{table}
Unfortunately, for $n > 4$, the degree computation does not terminate.

Computing degree with {\tt Macaulay2}  can also provide improved estimates for particular cases.
For instance, for $n=2$ with the point masses located at $(1,0)$ and $(-1,0)$, we find that the number of equilibria is at most $52$.

  The {\tt Macaulay2} bound can be improved for
  larger examples by applying BKK theory, as was
  employed by \cite{KuRoSm}.
  We use the software {\tt PHCpack} \cite{PHCpack} to compute the mixed volume $MV_n$
  of the Newton polytopes for the $n+1$-body system described in Equations \eqref{eq:append}.
  Degree computation in {\tt Macaulay2} is performed using the symbolic Gr\"{o}bner basis
  algorithm, while {\tt PHCpack} computes the mixed volume primarily through linear
  optimization. For this reason, it is able to handle much larger examples.

  Bernstein's Theorem \cite{Bernstein} implies that the mixed volume is an upper bound
  for the number of solutions to the system in the algebraic torus $\mathbb{C}_*^n$. We must be cautious
  because our system can have solutions with $x$ or $y$ equal to $0$; however, a generic rotation of the plane preserves
  the number of solutions while removing them (except the origin) from the $x$- and $y$-axes.
  As for the $w_i$, these are guaranteed to be nonzero by the equations. So, we can assume that all
  but one of our solutions is in the torus, thus the mixed volume plus one is an upper bound for us.

  We calculated $MV_n$ for $n$ up to $50$. The $n = 50$ case took only 3
  seconds to compute, but {\tt PHCpack} starts making integer overflow errors.
  For all integers $1 \leq n \leq 50$, $MV_n$ agrees with the
  function $f(n) = (9n^2 + 3n - 4)2^{n-1}$. We obtained this formula by first factoring
  out the exponential term, then matching the remaining sequence to {\tt A283394}
  in the OEIS \cite{oeis}.

  In an attempt to further simplify the form of the Newton polytopes, we can introduce
  even more new variables: $a_i$ and $b_i$ to denote the differences $x-x_i$ and $y-y_i$,
  respectively. After adding these variables in, we can remove $x$ and $y$ since they are equal
  to $x_i + a_i$ and $y_i + b_i$ for any $i$. The resulting system is below:

  \be\label{eq:append3}
\begin{cases}
\displaystyle (a_1 + x_1) - \sum_{i=1}^n m_i a_i w_i^3 = 0 \\
\displaystyle
(b_1 + y_1) -\sum_{i=1}^n m_ib_iw_i^3  =0 \\
w_i^2(a_i^2+b_i^2)-1=0, & i = 1,2,\ldots,n \\
a_i = a_1 + (x_1-x_i), & i = 2,3,\ldots,n \\
b_i = b_1 + (y_1-y_i), & i = 2,3,\ldots,n \\
\end{cases}.
\ee

This is a system of $ 2+ n +(n-1)+(n-1) = 3n$ equations in $3n$ variables; recall that $m_i,x_i$, and $y_i$ are
fixed parameters. Here, too, we computed the sequence of mixed volumes $\widetilde{MV}_n$ of these Newton polytopes
up to $n = 50$. Because of the larger number of variables, the computation took longer: the $n=50$ case took 2 minutes
and 17 seconds to compute. For the first 50 values, $\widetilde{MV}_n$ agrees with the function $g(n) = (9n^2 + n + 2)2^{n-1}$ 
(again using OEIS and finding {\tt A006137}). This is a better bound than $MV_n$ for all $n > 3$.

\begin{conj}
  The sequence of polynomial systems defined by \eqref{eq:append3} has mixed volume $\widetilde{MV}_n = (9n^2 + n + 2)2^{n-1}$.
  As a corollary, the maximum number of solutions $N(n)$ would satisfy
  \[N(n) \leq (9n^2 + n + 2)2^{n-1} + 1.\]
  \end{conj}


\subsection{A Lower Bound on the Maximum Number of Equilibria}


\begin{thm}\label{thm:5n-5}
The maximum number $N(n)$ of solutions to \eqref{eq:FxFy} satisfies the lower bound  $$N(n) \geq 5n-5 $$ for all $n \geq 2$.
\end{thm}

The proof of the theorem will analyze the so-called ``ring configuration'' from the study of the circular-restricted $(n+1)$-body problem \cite{Kalvouridis1998}.
The ring configuration was first introduced by Maxwell while studying the stability of Saturn's rings \cite{CroustKal2011}---for Maxwell the mass of each peripheral primary is taken to be small while the mass of the central primary is large, whereas we will need to take the mass of the peripheral primaries to be relatively large in order to produce the $5n-5$ equilibria.

In the proof of Theorem \ref{thm:5n-5} will need the following elementary lemma.

\begin{lemma}\label{lemma:trig}
For $n > 2$ an integer, $\displaystyle \sum_{k = 0}^{n-1} \cos{\left( \frac{2(2k+1) \pi}{n}\right)} = 0$.
\end{lemma}
\begin{proof}[Proof of Lemma \ref{lemma:trig}]
The left hand side can be reduced to a telescoping sum using the trigonometric product-to-difference identity as follows.
\begin{align*}
    \sum_{k = 0}^{n-1} \cos \left(\frac{2(2k+1) \pi}{n} \right) & = \frac{1}{\sin(2\pi/n)} \sum_{k = 0}^{n-1} \cos \left(\frac{2(2k+1) \pi}{n} \right) \sin (2\pi/n) \\
& = \frac{1}{\sin(2\pi/n)} \sum_{k = 0}^{n-1} \left[ \sin \left(\frac{4(k+1) \pi}{n} \right) -  \sin \left(\frac{4k\pi}{n} \right) \right] \\
& = \frac{1}{\sin(2\pi/n)} \left( \sin \left(4 \pi\right) -  \sin 0 \right) = 0.
\end{align*}
\end{proof}

\begin{proof}[Proof of Theorem \ref{thm:5n-5}]
We first address the cases $n=2$ and $n=3$ by referring to known examples.
For $n=2$, we note that $N(2) \geq 5$ follows from the classical result that the circular restricted three-body problem admits five equilibria (the famous five Lagrange points).
For $n=3$, we recall that the circular restricted four-body problem, with three primary masses at the vertices of an equilateral triangle, admits as many as ten equilibria, showing $N(3) \geq 10$ \cite{BaLe}.

For $n>3$ we consider the ring configuration \cite{Kalvouridis1998} in the circular-restricted $(n+1)$-body problem with $n$ primaries consisting of one primary at the origin with unit mass, and $n-1$ primaries located at (for $k=0,1,2,...,n-2$) the points  $\left(\cos{\frac{2k\pi}{n-1}},\sin{\frac{2k\pi}{n-1}}\right)$ with a common mass of $m$.
We will show that this configuration admits (at least) $5n-5$ equilibria, which agrees with numerical evidence from \cite{Kalvouridis1998}.

This configuration has $(n-1)$-fold symmetry.  We will restrict attention to two familes of rays each going from the origin to infinity: $n-1$ \emph{rays of type A} with angle $\frac{2k\pi}{n-1}$ for $k=0,1,2,...,n-2$ and $n-1$ \emph{rays of type B} with angle $\frac{(2k+1)\pi}{n}$ for $k=0,1,2,...,n-2$.

\medskip

\noindent{\bf Claim 1.} There are two equilibria on each of the rays of type A.

\medskip

\noindent{\bf Claim 2.}  For $m$ sufficiently large, there are three equilibria on each of the rays of type B. 

\medskip

\begin{proof}[Proof of Claim 1]
In view of the symmetry in the problem, it is suffices to verify the statement for one of the rays of type A, say the positive $x$-axis.
The vertical component of force vanishes along the $x$-axis, i.e., $F_y(x,0)=0$.
To elaborate, the point masses located off the $x$-axis can be paired with one another by reflection over the $x$-axis.
From geometric considerations, the vertical components of force cancel within each pair.

Hence, to locate equilibria on the positive $x$-axis, it is enough to find solutions of $F_x(x,0)=0$.  Upon inspection of $F_x(x,0)$, we notice the following limiting behavior
$$\lim_{x\to 0^+} F_x=-\infty, \quad \lim_{x\to 1^-} F_x=\infty, \quad \lim_{x\to 1^+}F_x=-\infty, \quad \text{and } \lim_{x\to \infty}F_x=\infty.$$
Thus, there is a sign change in each of the intervals $(0,1)$ and $(1,\infty)$, so by the intermediate value theorem there is a point at which $F_x=0$ in each of those two intervals. 
This shows that there are two equilibria on each of the rays of type A, as stated in the claim.
\end{proof}

\begin{proof}[Proof of Claim 2]
Considering the symmetry in the problem, it suffices to verify the statement for one ray of type B.
By a rotation of the plane, we move this ray to the positive $x$-axis, and the point masses are now positioned at 
$\left(\cos{\frac{(2k+1)\pi}{n-1}},\sin{\frac{(2k+1)\pi}{n-1}}\right)$.
As before, the vertical component of force vanishes on the $x$-axis, i.e., $F_y(x,0)=0$.
Thus, to locate equilibria, it is enough to find solutions of the equation $F_x(x,0)=0$.
We have, for $x>0$,
\be
F_x(x,0) = x - \frac{1}{x^2} + m \cdot g(x),
\ee
where 
$$g(x)=\sum_{k = 0}^{n-2} \frac{\cos{\frac{(2k+1) \pi}{n-1}}-x}{[x^2-2x\cos{\frac{(2k+1) \pi}{n-1}}+1]^{3/2}}.$$

Differentiating $g(x)$ with respect to $x$, we find
$$g'(x)=\sum_{k=0}^{n-2}-\frac{3(2x-2\cos{\frac{(2k+1)\pi}{n-1}})(-x+\cos{\frac{(2k+1)\pi}{n-1}})}{2(-2x\cos{{\frac{(2k+1)\pi}{n-1}}{}}+x^2+1)^{5/2}}-\frac{1}{(-2x\cos{\frac{(2k+1)\pi}{n-1}}+x^2+1)^{3/2}}.$$
Evaluating this at $x=0$, we obtain
\begin{align*}
    g'(0) &= \sum_{k =0}^{n-2} -1+3 \cos^2 \frac{(2k+1) \pi}{n-1} \\
    &= \sum_{k =0}^{n-2} -\sin^2{\frac{(2k+1) \pi}{n-1}}+2\cos^2{\frac{(2k+1) \pi}{n-1}} \\
    &= \sum_{k =0}^{n-2} \cos{\frac{2(2k+1) \pi}{n-1}}+\cos^2{\frac{(2k+1) \pi}{n-1}} \\
    &= \sum_{k =0}^{n-2} \cos^2{\frac{(2k+1) \pi}{n-1}},
\end{align*}
where we have used Lemma \ref{lemma:trig} in the last step. From this we see that $g'(0)>0$, and this implies that for $\delta>0$ sufficiently small we have $g(\delta) > 0$.  Then let $m>0$ be large enough that $m g(\delta) >|\delta-\frac{1}{\delta^2}|$. Since $F_x(x,0) = x - \frac{1}{x^2} + m g(x)$, this implies $F_x(\delta,0)>0$.

On the other hand, we have $g(1) < 0$, since the summand satisfies, for every $k$, 
$$\frac{\cos{(\frac{(2k+1) \pi}{n-1})}-1}{[2-2\cos{(\frac{(2k+1) \pi}{n-1})}]^{3/2}} < 0.$$  This implies $F_x(1,0) = 1 - 1 + g(1) < 0$.

Then, by the intermediate value theorem, there is an equilibrium point on the $x$-axis between $x=\delta$ and $x=1$.
Moreover, we have the following limiting behavior 
$$\lim_{x\to 0^+} F_x(x,0)=-\infty, \quad \lim_{x\to \infty} F_x(x,0)=\infty. $$ So there is an equilibrium point in the interval $0< x < \delta$ as well as an equilibrium point in the interval $x>1$.  Hence, there are three points of equilibrium on the $x$-axis.  This proves the claim.
\end{proof}

As we have two points of equilibrium on $n-1$ rays by Claim 1 and, assuming $m$ is sufficiently large, three points of equilibrium on $n-1$ rays by Claim 2, this shows that there are a total of at least $5n-5$ points of equilibrium. This proves the lower bound $N(n) \geq 5n-5$  stated in the theorem.
\end{proof}

\begin{figure}[t]
\centering
\includegraphics[width=0.47\linewidth]{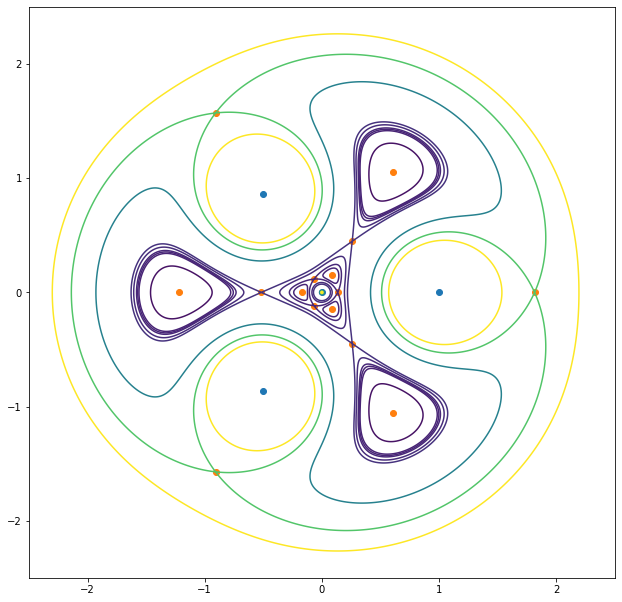}
\includegraphics[width=0.47\linewidth]{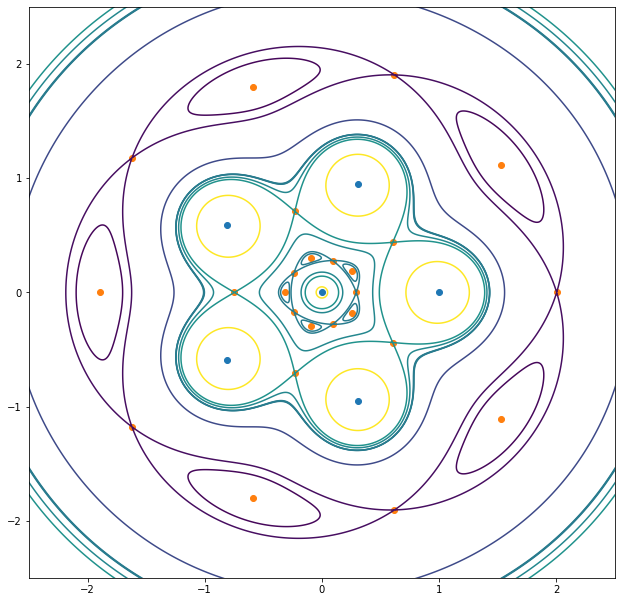}
\caption{Level curves of the potential for the ring configuration with $n=3+1$ point masses (left) and with $n=5+1$ point masses (right). In the $n=3+1$ case, the peripheral masses each have mass $1$ and the small central mass has mass $1/100$.  In the $n=5+1$ case, the peripheral masses each have mass $7/6$ and the small central mass has mass $7/60$.}
\label{fig:Ring4}
\end{figure}


\section{Concluding Remarks}

Our results show that $5n-5 \leq N(n) \leq 4^{n+2}$.  We believe the upper bound is far from sharp, and we conjecture that $N(n)$ grows linearly.
In light of the result of Khavinson and Neumann \cite{KhNe} and Remark \ref{rmk:ring}, the following question concerning Problem \ref{prob:main} seems natural

\begin{q}\label{q:main}
Does the maximum number of equilibria for Problem \ref{prob:main} satisfy 
$\displaystyle N(n)=5n-5$?
\end{q}

Since the examples used in the proof of Theorem \ref{thm:5n-5} consisted of central configurations, a positive answer to Question \ref{q:main} would be interesting, as it would mean that the maximum number $N(n)$ is attained within the small subset of parameter space that is relevant to the circular-restricted $(n+1)$-body problem.

We have defined $N(n)$ as the maximum number of equilibria while taking the maximum over the generic set of parameters for which the number of equilibria is finite.  It seems likely that the number of equilibria is finite in general for $n \geq 2$, but as we have mentioned above it is an open problem to show this.

The problem of establishing general finiteness results in Maxwell's problem is also open and described as ``very irritating'' by B. Shapiro \cite{Sh}
(it is open even in the case that the charges have a common sign).

Finiteness in the gravitational lensing problem can be shown using complex variables and a rigidity theorem for the so-called Schwarz function of an analytic curve \cite[Introduction]{KhNe}: finiteness of the number of equilibria holds except for $n=1$ and $z_1=w$ (where a circle of equilibria produces an ``Einstein ring'').

Concerning the proof of Theorem \ref{thm:UB}, in the study of Problem \ref{prob:Max}, the planar restriction of Maxwell's problem, an alternative method has been used in \cite{Killian} for approaching a system very similar to (\ref{eq:main}) without introducing additional variables.  The method involves repeatedly squaring, rearranging, and clearing denominators in order to remove the square roots in each of the expressions
$ |z-z_i| = \sqrt{(x-x_i)^2 + (y-y_i)^2}$.
After this process results in a polynomial system, the classical bivariate Bezout's theorem can be applied.
However, applying this method to Problem \ref{prob:main} gives the upper bound $4^n(3n+1)^2$ which has an additional quadratic factor as compared to the upper bound we obtained in Theorem \ref{thm:UB}.

In the opposing direction, we may apply the method of proof of Theorem \ref{thm:UB} to Problem \ref{prob:Max}.
With very minor adaptations to the proof of Theorem \ref{thm:UB}, the result is that the number of equilibria for the planar restriction of Maxwell's problem is, in the generic case, at most $4^{n+2}$. When $n \geq 4$, this gives an improvement on the upper bound $4^{n-2}(3n-2)^2$ obtained in \cite{Killian} from using the method of repeated squaring described above.

\noindent {\bf Data availability statement.}
The data that support this study are available upon request.

\noindent {\bf Acknowledgements.}
The third author acknowledges support from the Simons Foundation under the grant 712397.

\bibliographystyle{alpha}
\bibliography{prob}

\vspace{0.08in}

{\em

Nickolas Arustamyan\\
email: narustamyan2017@fau.edu
}

\vspace{0.08in}

{\em

Christopher Cox\\
email: ccox2017@fau.edu
}

\vspace{0.08in}

{\em
Erik Lundberg\\
email: elundber@fau.edu
}

\vspace{0.08in}

{\em

Sean Perry\\
email: sperry9@fau.edu
}

\vspace{0.08in}

{\em

Zvi Rosen\\
email: rosenz@fau.edu
}

\end{document}